\documentclass{eptcs}

\usepackage{amsmath}
\usepackage{amsfonts}
\usepackage{amssymb}
\usepackage{amsxtra}
\usepackage{amsthm} 
\usepackage{url}
\usepackage{graphicx,color}
\usepackage{float}
\usepackage{stfloats}
\usepackage{dsfont}
\usepackage{hyperref}
\usepackage{mathrsfs}
\usepackage{array}
\usepackage{rotating}
\usepackage{calc}
\usepackage{caption}
\usepackage{subcaption}
\usepackage{cancel}
\usepackage{lipsum}
\usepackage{anyfontsize} 

\newtheorem{theorem}{Theorem}
\newtheorem{definition}{Definition}

\newtheorem{remark}{Remark}
\newtheorem{corollary}{Corollary}
\newtheorem{example}{Example}
\newtheorem{lemma}{Lemma}

\newcommand{\tuple}[1]{\langle #1 \rangle}

\newcommand{\lrsim}{\overset{\rightarrow}{\rule{5.8pt}{0.5pt}}}

\newcommand{\bisim}{\overset{\leftrightarrow}{\rule{5.8pt}{0.5pt}}}
\newcommand{\ghmldiamond}[1]{\langle\negthinspace\langle #1 \hspace{1pt} \rangle\negthinspace\rangle}
\newcommand{\tightoverset}[2]{%
  \mathop{#2}\limits^{\vbox to -.5ex{\kern-0.6ex\hbox{$#1$}\vss}}}
\newcommand{\leftendpoint}[1]{{\tightoverset{\leftharpoonup}{{#1}_{~}\negthinspace}\negthinspace}}
\newcommand{\rightendpoint}[1]{{\tightoverset{\rightharpoonup}{{#1}_{~}\negthinspace}\negthinspace}}

\title{Bisimulation and Hennessy-Milner Logic for Generalized Synchronization Trees\footnote{This work was supported by NSF grant CNS-1446665.}}
\newcommand{\titlerunning}{Bisimulation and Hennessy-Milner Logic for GSTs} 

\author{James Ferlez\textsuperscript{1,2}, Rance Cleaveland\textsuperscript{1,3} and Steve Marcus\textsuperscript{1,2} \\
\textsuperscript{\tiny 1} {\small \emph{The Institute for Systems Research, University of Maryland, College Park}}\\
\textsuperscript{\tiny 2} {\small \emph{Department of Electrical and Computer Engineering, University of Maryland, College Park}} \\
\textsuperscript{\tiny 3} {\small \emph{Department of Computer Science, University of Maryland, College Park}}}
\newcommand{\authorrunning}{J. Ferlez, R. Cleaveland and S. Marcus}

\date{\today}

\begin{document}

\maketitle %

\begin{abstract}
In this work, we develop a generalization of Hennessy-Milner Logic (HML) for 
Generalized Synchronization Trees (GSTs) that we call Generalized Hennessy 
Milner Logic (GHML). Importantly, this logic suggests a strong relationship 
between (weak) bisimulation for GSTs and ordinary bisimulation for 
Synchronization Trees (STs). We demonstrate that this relationship can be used 
to define the GST analog for image-finiteness of STs. Furthermore, we 
demonstrate that certain maximal Hennessy-Milner classes of STs have 
counterparts in maximal Hennessy-Milner classes of GSTs with respect to GST 
weak bisimulation. We also exhibit some interesting characteristics of these 
maximal Hennessy-Milner classes of GSTs. 
\end{abstract}

\section{Introduction} 
\label{sec:introduction}
	In the context of discrete systems modeled as Synchronization Trees (STs), 
	Hennessy and Milner first 
	noticed a relationship between bisimulation and a simple modal logic, 
	subsequently to be known as Hennessy-Milner logic (HML) 
	\cite{hennessy_algebraic_1985}. In particular, they observed that HML 
	characterizes bisimulation within the class of image-finite STs in the 
	following sense: two image-finite STs are bisimilar if and only if they 
	satisfy exactly the same HML formulas.
	Subsequent to Hennessy and Milner's original work, HML has likewise been 
	shown to characterize bisimulation within other classes of STs (though not 
	the class of \emph{all} STs \cite{milner_handbook_1990}). Indeed, any class 
	of STs for which modal equivalence implies bisimulation is known as a 
	\emph{Hennessy-Milner class} (HM class), and a number of maximal ones have 
	been exhibited \cite{hollenberg_hennessy-milner_1995}. 

	Such characterizations of bisimulation have significant ramifications for 
	the verification of system properties: if two systems belong to the same HM 
	class, then they can be checked for bisimulation equivalence by checking 
	HML formulas instead. In addition, if two such systems are \emph{not} 
	bisimilar, then an HML formula can bear witness to this lack of 
	bisimilarity~\cite{cleaveland1990automatically}.  For a simple logic such 
	as HML, this is a considerable advantage. Moreover, the existence of 
	\emph{maximal} HM classes is particularly important in the study of STs and 
	process algebra, given the inherent compositionality of those objects. In 
	particular, it is useful to know which operations preserve membership in a 
	maximal HM class; this has been considered for some cases in 
	\cite{hollenberg_hennessy-milner_1995}. 

	Recently, the authors have proposed Generalized Synchronization Trees 
	(GSTs) \cite{ferlez_generalized_2014} as a flexible modeling framework for 
	\emph{non-discrete} systems such as continuous or hybrid systems. Despite 
	their broader applicability, GSTs have many similarities to STs: elegant 
	composition operators between GSTs are plentiful, and there are well 
	defined notions of bisimulation (see \cite{ferlez_generalized_2014}). Thus, 
	GSTs are natural candidates for the treatment of both HML-like logics and 
	HM classes, especially with a view to studying the composition of 
	continuous and hybrid systems. This paper launches such a study and makes 
	three crucial contributions on the topic of modal logic for non-discrete 
	systems: first, we define a novel generalization of HML that has semantic 
	parity with trajectories in GSTs; second, we use this logic to define a 
	notion of image finiteness for GSTs; and third, we exhibit a partial 
	characterization of maximal HM classes in the context of our generalized 
	HML. The third contribution is particularly novel, since there seem to be 
	no results at all about \emph{maximal} HM classes for generic hybrid system 
	models, much less results in a framework as flexible and compositional as 
	GSTs.  Since GSTs can exhibit infinite -- and even continuous -- 
	non-determinism, they offer a particularly rich setting in which to explore 
	the structure of HM classes. 

	There are other results in the hybrid systems literature that relate 
	bisimulation to modal logics (see e.g. \cite{davoren_simulations_2007}), 
	but these results typically focus on establishing that \emph{bisimulation 
	preserves the satisfaction of formulas} from some modal logic. Almost none 
	consider the problem of specifying when modal equivalence implies 
	bisimilarity, i.e. the identification of Hennessy-Milner classes. The few 
	papers that do consider the problem of identifying HM classes seem to be 
	concerned with probabilistic systems: see 
	\cite{desharnais_bisimulation_2002, doberkat_stochastic_2005, 
	doberkat_hennessymilner_2007, terraf_bisimilarity_2015, 
	doberkat_stochastic_2017, dargenio_bisimulations_2012, budde_theory_2014} 
	for example. However, in most of these papers, the following quote is 
	emblematic of the source of these results: ``... the probabilistic systems 
	we are considering, without explicit nondeterminism, resemble deterministic 
	systems quite closely, rather than nondeterministic systems'' 
	\cite{desharnais_bisimulation_2002}. In other words, these papers typically 
	end up with something like an image-finite assumption, and that forms the 
	basis of their HM classes.  On the other hand, the papers that do not have 
	an image-finiteness assumption (\cite{terraf_bisimilarity_2015, 
	doberkat_stochastic_2017}) always consider more complicated logics than 
	HML, and do not address questions of maximal HM classes. 
\section{Background} 
\label{sec:background}
	This section describes several foundational results that will be used 
	subsequently. The first subsection contains background material on GSTs, 
	including a review of the relevant notions of bisimulation for the same. 
	The second subsection contains background material on HML and Hennessy and 
	Milner's result for image-finite processes. The third subsection describes 
	some maximal Hennessy-Milner classes over Kripke structures 
	\cite{hollenberg_hennessy-milner_1995} and some necessary preliminaries on 
	the canonical model \cite{goldblatt_logics_1992}. Table 
	\ref{table:main_notation} describes some common notation that will be used 
	throughout this section and the rest of the paper.

	\begin{table}[t]
		\begin{center}
		\begin{tabular}{ |ccl| }
			\hline
			$ Sys_1 \; \lrsim_{X} \; Sys_2$ & : &  $Sys_2$ \emph{simulates} $Sys_1$ (w.r.t. simulation notion $X$) \\[2pt]
			$ Sys_1 \; \bisim_{X} \; Sys_2$ & : &  $Sys_1$ and $Sys_2$ are bisimilar (w.r.t. simulation notion $X$) \\[2pt]
			$ s \; \bisim_{X} \; t$ & : &  \emph{States (or worlds, nodes)} $s$ and $t$ are bisimilar (w.r.t. simulation notion $X$) \\[2pt]
			$Sys_1 \; \approx_{Y} \; Sys_2$ & : & $Sys_1$ and $Sys_2$ satisfy the same formulas of logic $Y$ \\[2pt]
			$s \; \approx_{Y} \; t$ & : & \emph{States (or worlds, nodes)} $s$ and $t$ satisfy the same formulas of logic $Y$ \\[2pt]
			\hline
		\end{tabular}
		\end{center}
		\caption{Notation for simulation, bisimulation and modal equivalence.}
		\label{table:main_notation}
	\end{table}
	\subsection{Generalized Synchronization Trees} 
	\label{sub:generalized_synchronization_trees}
	This section summarizes the theory of GSTs as presented in 
	\cite{ferlez_generalized_2014}; the interested reader is referred to that 
	reference for more details.

	GSTs extend Milner's Synchronization Trees (STs) via a generalized notion 
	of tree. In particular, the essential element of a GST is the \emph{tree 
	partial order}, a well known structure in the mathematical literature (see 
	\cite{jech_set_2003}, for example). 
	\begin{definition}[Tree \cite{jech_set_2003,ferlez_generalized_2014}]
	\label{def:tree}
	A \emph{tree} is a triple $\tuple{P, \preceq, p_0}$, where $\tuple{P, 
	\preceq}$ is a partial order, $p_0 \in P$, and the following  hold.
	\begin{enumerate}
		\item $p_0 \preceq p$ for all $p \in P$ ($p_0$ is called the \emph{root} 
			of the tree)
		\item For any $p \in P$ the set $[p_0, p] \triangleq \{p^\prime \in 
			P : p_0 \preceq p^\prime \preceq p \}$ is totally ordered by 
			$\preceq$.
	\end{enumerate}
	\end{definition}
	In this definition of a tree, there is no inherent notion of edge, or 
	``discrete'' transition, so unlike STs external interactivity cannot be 
	captured by labeling edges; such action labels need a different encoding. 
	The definition of GSTs below suggests just such a scheme.   
	\begin{definition}[Generalized Synchronization Tree \cite{ferlez_generalized_2014}]
	\label{definition:gst}
	Let $L$ be a set of labels.  Then a \textbf{Generalized Synchronization 
	Tree (GST)} is a tuple $\tuple{P,\preceq,p_0,\mathcal{L}}$, where:
	\begin{enumerate}
		\item $\tuple{P,\preceq,p_0}$ is a tree in the sense of Definition 
			\ref{def:tree}; and
		\item $\mathcal{L} \in P\backslash\{p_0\} \rightarrow L$ is the labeling 
			function (may be partial).
	\end{enumerate}
	\end{definition}
	Intuitively, labels in a GST are affixed to nodes in the tree.  If the tree 
	is discrete, it can be converted into an ST by moving the node labels onto 
	the incoming edge from the nodes parent (note that roots are not labeled in 
	GSTs, so this operation is well-defined). 

	Another consequence of a lack of discrete transitions is that bisimulation 
	must be defined differently than for STs. Specifically, bisimulation 
	between GSTs is defined in \emph{trajectories}, which are totally ordered 
	sets of nodes in a GST that play roughly the same role as transitions (or 
	sequences of transitions) in discrete bisimulation.
	\begin{definition}[Trajectory \cite{ferlez_generalized_2014}]
	\label{def:trajectory}
	Let $\tuple{P,\preceq,p_0,\mathcal{L}}$ be a GST, and let $p \in P$.  Then 
	a \textbf{trajectory} from $p$ is either:
	\begin{enumerate}
		\item the set $(p,p'] \triangleq \{p'' \in P : p \prec 
			p'' \preceq p' \}$ for some 
			$p' \succ p$, or
		\item a (set-theoretic) maximal linear subset $P' \subseteq P$ with 
			the property that for all $p' \in P'$, $p' \succ p$.
	\end{enumerate}
	Trajectories of the first type are called \textbf{bounded}, while those of 
	the second type are called \textbf{(potentially) unbounded}. 
	\end{definition}

	To account for the labels on the nodes of a trajectory, we define a 
	notion of order equivalence to parallel the notion of 
	identically labeled transitions in a ST: 
	\begin{definition}[Order Equivalence \cite{ferlez_generalized_2014}]
	\label{def:order-equivalence}
	Let $\tuple{P, \preceq_P, p_{0}, \mathcal{L}_P}$ and 
	$\tuple{Q,\preceq_Q,q_{0},\mathcal{L}_Q}$ be GSTs, let $T_p, T_q$ be 
	trajectories from $p \in P$ and $q \in Q$ respectively.  Then $T_p$ and 
	$T_q$ are \textbf{order-equivalent} if there exists a bijection $\lambda 
	\in T_P \rightarrow T_Q$ such that: 
	\begin{enumerate}
		\item $p_1 \preceq_P p_2$ if and only if $\lambda(p_1) \preceq_Q 
			\lambda(p_2)$ for all $p_1, p_2 \in T_P$, and  
		\item $\mathcal{L}_P(p) = \mathcal{L}_Q(\lambda(p))$ for all $p \in 
			T_P$ 
	\end{enumerate}
	When $\lambda$ has this property, we say that $\lambda$ is an \emph{order 
	equivalence} from $T_P$ to $T_Q$. 
	\end{definition}

	We now recall the two notions of 
	simulation from \cite{ferlez_generalized_2014}; corresponding notions of 
	bisimulation can be defined in the obvious way 
	\cite{ferlez_generalized_2014}.
	\begin{definition}[Weak Simulation for GSTs\footnote{``Weak'' is used here only as a relative term; it does not refer to the inclusion of $\tau$ transitions.} \cite{ferlez_generalized_2014}]
	\label{def:weak-simulation}
	Let $G_1 = \tuple{P, \preceq_P, p_{0}, \mathcal{L}_P}$ and $G_2 = 
	\tuple{Q,\preceq_Q,q_{0},\mathcal{L}_Q}$ be GSTs.  Then $R \subseteq P 
	\times Q$ is a \textbf{weak simulation from $G_1$ to $G_2$} if, whenever 
	$\tuple{p,q} \in R$ and $p' \succeq p$, then there is a $q' \succeq q$ such 
	that:
	\begin{enumerate}
	\item $\tuple{p',q'} \in R$, and
	\item Trajectories $(p,p']$ and $(q,q']$ are order-equivalent.
	\end{enumerate}
	We say $G_1 \lrsim_w G_2$ if there is a weak simulation $R$ from $G_1$ to 
	$G_2$ with $\tuple{p_0,q_0} \in R$. 
	\end{definition}
	\begin{definition}[Strong Simulation for GSTs \cite{ferlez_generalized_2014}]
	\label{def:strong-simulation}
	Let $G_1 = \tuple{P, \preceq_P, p_{0}, \mathcal{L}_P}$ and $G_2 = 
	\tuple{Q,\preceq_Q,q_{0},\mathcal{L}_Q}$ be GSTs.  Then $R \subseteq P 
	\times Q$ is a \textbf{strong simulation from $G_1$ to $G_2$} if, whenever 
	$\tuple{p,q} \in R$ and $T_p$ is a trajectory from $p$, there is a 
	trajectory $T_q$ from $q$ and bijection $\lambda \in T_p \rightarrow T_q$ 
	such that: 
	\begin{enumerate}
		\item $\lambda$ is an order equivalence from $T_p$ to $T_q$, and
		\item $\tuple{p',\lambda(p')} \in R$ for all $p' \in T_p$.
	\end{enumerate}
	We write $G_1 \lrsim_s G_2$ if there is a strong simulation $R$ from $G_1$ 
	to $G_2$ with $\tuple{p_0,q_0} \in R$.  
	\end{definition}
	%
	\subsection{Hennessy-Milner Logic and a Hennessy-Milner Class} 
	\label{sub:hennessy_milner_logic_and_a_hennessy_milner_class}
		\subsubsection{Hennessy-Milner Logic} 
		\label{ssub:hennessy_milner_logic}
		Hennessy-Milner Logic is defined as follows; note the lack of 
		atomic propositions.
		\begin{definition}[Hennessy-Milner Logic (HML) \cite{hennessy_algebraic_1985}]
		\label{def:hml_logic} 
		Given a set of labels, $L$, \textbf{Hennessy-Milner Logic (HML)} is the 
		set of formulas $\Phi_{\text{HML}}(L)$ specified as follows, where 
		$\ell \in L$: 
		\begin{equation}
			\varphi := \top \quad | 
				\quad \neg \varphi \quad | 
				\quad \varphi_1 \wedge \varphi_2 \quad | 
				\quad \langle \ell \rangle \varphi.
		\end{equation}
		\end{definition}
		In \cite{hennessy_algebraic_1985}, the semantics of this logic are 
		defined in following familiar way over STs.
		\begin{definition}[Semantics of HML \cite{hennessy_algebraic_1985}]
		\label{def:semantics_of_hml} 
		A satisfaction relation of HML for a set of STs $\mathcal{P}$ is a set 
		$\models \subseteq \mathcal{P} \times \Phi_{\text{HML}}(L)$ such that:
		\begin{enumerate}
			\item $p \models \top$ for all $p \in \mathcal{P}$;

			\item $p \models \neg \varphi$ if and only if $p ~ 
				\cancel{\models} ~ \varphi$;

			\item $p \models \varphi_1 \wedge \varphi_2$ if and only if $p 
				\models \varphi_1$ and $p \models \varphi_2$; and

			\item $p \models \langle l \rangle \varphi$ if and only if there 
				exists a $p^\prime$ such that $ p \overset{l}{\rightarrow} 
				p^\prime$ and $p^\prime \models \varphi$.
		\end{enumerate}
		Here $p \overset{l}{\rightarrow} p'$ means $p'$ is a subtree of $p$ 
		whose root is a child of the root of $p$, with the edge connecting the 
		roots labeled by $l$. 
		\end{definition}

		\begin{remark}
			We will freely avail ourselves of usual derived operators such as 
			$\bot$, $\vee$, $\rightarrow$ and $[\ell]$.  
		\end{remark}
		\subsubsection{Bisimulation and a Hennessy-Milner Class} 
		\label{ssub:bisimulation_and_a_hennessy_milner_class}
		Hennessy and Milner noticed that STs satisfying the same HML 
		formulas need not be bisimilar; i.e. $p \approx_{\text{HML}} q \;\; 
		\cancel{\Rightarrow} \;\; p \; \bisim \; q$ in general 
		\cite{hennessy_algebraic_1985}. Nevertheless, they exhibited a class of 
		STs for which HML modal equivalence \emph{does} imply 
		bisimilarity: that is the class of \emph{image-finite STs}.

		\begin{definition}[Image-Finite Process \cite{hennessy_algebraic_1985}]
		\label{def:image_finite_st}
		A ST $p \in \mathcal{P}$ is said to be image-finite if for each 
		subtree $q$ of $p$ (including $p$ itself) and 
		each label $\ell \in \mathcal{L}$, the set $\{q^\prime : q 
		\overset{\ell}{\rightarrow} q^\prime\}$ is finite.
		\end{definition}
		Hennessy and Milner proved the following theorem.
		\begin{theorem}[Image-Finite Hennessy-Milner Theorem \cite{hennessy_algebraic_1985}]
		\label{thm:hm_theorem_image_finite_sts}
		Let $p$ and $q$ be any two image-finite STs. Then
		\begin{equation}
			\label{eq:hm_theorem_prototype}
			p \; \bisim \; q \Longleftrightarrow p \approx_{\text{HML}} q.
		\end{equation}
		\end{theorem}
		\begin{remark}
			Any theorem with a conclusion of 
			the form \eqref{eq:hm_theorem_prototype} is called a 
			\textbf{Hennessy-Milner Theorem}. Likewise, any class of STs 
			(or systems, Kripke structures, etc.) for which a Hennessy-Milner 
			Theorem can be exhibited is called a \textbf{Hennessy-Milner class}.
		\end{remark}
	%
	\subsection{The Canonical Model and (Maximal) Hennessy-Milner Classes} 
	\label{sub:the_canonical_model_and_maximal_hm_classes}
	Since Hennessy and Milner's work \cite{hennessy_algebraic_1985}, other 
	\emph{Hennessy-Milner classes} have been exhibited. In particular, Visser, 
	via Hollenberg \cite{hollenberg_hennessy-milner_1995}, has 
	generalized the idea of a Hennessy-Milner class to Kripke structures, and 
	exhibited certain \emph{maximal} Hennessy-Milner classes of Kripke 
	structures. This section describes the characterization of these classes.

	As a prelude, we introduce the following familiar definitions of modal 
	logic, Kripke structure and bisimulation between Kripke structures.

	\begin{definition}[Modal Logic \cite{goldblatt_logics_1992}]
	\label{def:modal_logic}
	By a \textbf{modal logic}, we mean formulas constructed as in Definition 
	\ref{def:hml_logic} but with the addition of propositional variables 
	(atomic propositions). The set of formulas with a set of 
	propositional variables $\Theta$ and a set of labels $L$ is denoted by 
	$\Phi_{\Theta}(L)$.
	\end{definition}
	\begin{definition}[Kripke structure \cite{goldblatt_logics_1992,hollenberg_hennessy-milner_1995}]
	\label{def:kripke_structure}
	A Kripke structure of a set of labels $L$ and a set of 
	propositional variables $\Theta$ is a tuple $\mathbf{S} = (S, \{R_\ell 
	\subseteq S \times S : \ell \in L \}, V)$ where
	\begin{itemize}
		\item $S$ is the set of states (or worlds);

		\item for each $\ell \in L$, $R_\ell \subseteq S \times S$ 
			is a transition relation for label $\ell$; and

		\item $V : \Theta \rightarrow 2^S$ is a function that maps 
			propositional variables to sets of states.
	\end{itemize}
	\end{definition}
	\begin{definition}[Satisfaction relation for a Kripke structure]
	\label{def:sat_rel_kripke}
	Given a Kripke structure $\mathbf{S} = (S, \{R_\ell \subseteq S \times S : 
	\ell \in L \}, V) $, a satisfaction relation $\models \subseteq S \times 
	\Phi_{\Theta}(L)$ is defined as in Definition \ref{def:hml_logic} with the 
	addition that for any $\theta \in \Theta$, $s \models \theta$ if and only 
	if $s \in V(\theta)$.
	\end{definition}
	\begin{definition}[Bisimulation between Kripke structures]
	\label{def:bisim_kripke_structures}
	Given two Kripke structures $\mathbf{S}$ and $\mathbf{T}$,
	we say that $s \in S$ and $t \in T$ are bisimilar 
	or $s \; \bisim_{K} \; t$ if there is a bisimulation relation $\sim$ such 
	that $s \sim t$, and for all $s^\prime \in S, t^\prime \in T$, and $\theta 
	\in \Theta$, $s^\prime \sim t^\prime$ implies $s^\prime \models \theta 
	\Leftrightarrow t^\prime \models \theta$. Bisimulation between Kripke 
	structures is defined in the obvious way.
	\end{definition}
	\subsubsection{The Canonical Model for a Modal Logic} 
		\label{ssub:the_canonical_model_for_a_modal_logic}

		The so-called \emph{canonical model} -- called the \emph{Henkin model} 
		in \cite{hollenberg_hennessy-milner_1995} -- is a special Kripke 
		structure that is one of the most important tools in the study of 
		(normal) modal logic(s). In the canonical model
		states (or worlds) are defined in terms of the modal 
		formulas that they satisfy. In our context it has an important 
		connection to the maximal Hennessy-Milner classes we will consider; see 
		Subsection \ref{ssub:hennessy_milner_classes_for_kripke_structures}. 
		This subsection is meant to be a summary of the relevant material 
		in \cite{goldblatt_logics_1992}; the full details can be found therein.

		In order to define the canonical model, we must first establish certain 
		consistency criteria that we will enforce on any ``reasonable'' set of 
		formulas. Any such ``reasonable'' set of formulas will (somewhat 
		confusingly) be called a \emph{logic}.
		\begin{definition}[Logic \cite{goldblatt_logics_1992}]
		\label{def:logic}
		A set of modal formulas $\Lambda \subseteq \Phi_{\Theta}(L)$ 
		is called a \textbf{logic} if it satisfies:
		\begin{enumerate}
			\item $\Lambda$ contains all of the \textbf{tautologies}: that is 
				all of the formulas which are true irrespective of how we 
				assign truth values to modal sub-formulas and propositional 
				variables; and 

			\item $\Lambda$ is closed under modus ponens: if $\varphi_1 \in 
				\Lambda$ and $\varphi_1 \rightarrow \varphi_2 \in 
				\Lambda$\footnote{In terms of HML, $\varphi_1 \rightarrow 
				\varphi_2$ is shorthand for $\neg(\varphi_1 \wedge \neg 
				\varphi_2)$.}, then $\varphi_2 \in \Lambda$.
		\end{enumerate}
		\end{definition}
		\begin{definition}[$\Lambda$-Consistent Set of Modal Formulas \cite{goldblatt_logics_1992}]
		\label{def:consistent}
		Given a logic $\Lambda \subseteq \Phi_{\Theta}(L)$, a set of 
		modal formulas $\Gamma \subseteq \Phi_{\Theta}(L)$ is said to 
		be $\Lambda$\textbf{-consistent} if there is no formula of the form 
		$\varphi_0 \rightarrow ( \varphi_1 \rightarrow ( \dots \rightarrow ( 
		\varphi_n \rightarrow \bot) \dots )$ in $\Lambda$, where $\varphi_0, 
		\dots, \varphi_n \in \Gamma$.
		\end{definition}
		\begin{definition}[$\Lambda$-maximal set of formulas \cite{goldblatt_logics_1992}]
		\label{def:lambda_maximal_set}
		Given a logic $\Lambda$, a set of formulas $\Gamma \subseteq 
		\Phi_{\Theta}(L)$ is said to be \textbf{$\Lambda$-maximal} if 
		it satisfies the following two properties:
		\begin{enumerate}
			\item $\Gamma$ is $\Lambda$-consistent; and 

			\item for all $\varphi \in \Phi_{\Theta}(L)$, either $\varphi 
				\in \Gamma$ or $\neg\varphi \in \Gamma$.
		\end{enumerate}
		\end{definition}
		\begin{lemma}[Lindenbaum's Lemma \cite{goldblatt_logics_1992}]
		\label{lemma:lindenbaums_lemma}
		Given a logic $\Lambda$ and a $\Lambda$-consistent set of formulas 
		$\Gamma$, there exists a $\Lambda$-maximal set of formulas $\Gamma_0 
		\subseteq \Phi_{\Theta}(L)$  such that $\Gamma \subseteq 
		\Gamma_0$.
		\end{lemma}
		\begin{corollary}[The set of maximal $\Lambda$-consistent sets of formulas is non-empty \cite{goldblatt_logics_1992}]
		\label{cor:lambda_max}
		Given a logic $\Lambda$, let $S^\Lambda$ denote the set of 
		$\Lambda$-maximal sets of formulas. Then $S^\Lambda$ is non-empty.
		\end{corollary}

		For our purposes, Corollary \ref{cor:lambda_max} tells us that the set 
		of maximally consistent sets of formulas is non-empty, and hence, can 
		be used as the set of worlds for the canonical model. The next step in 
		the construction of the canonical model is to define the 
		\emph{transitions} between these states; these must ensure  that a 
		state -- a $\Lambda$-maximal set of formulas -- satisfies a modal 
		formula if an only if that formula is an element of the set. This 
		desirable property -- called the ``Henkin property'' in 
		\cite{hollenberg_hennessy-milner_1995} -- is the essence of the value 
		of the canonical model. It turns out that some additional conditions 
		must be imposed on a logic $\Lambda$ before transitions can be defined 
		between $\Lambda$-maximal sets of formulas in such a way that the 
		Henkin property holds. 

		\begin{definition}[Normal Logic \cite{goldblatt_logics_1992}]
		\label{def:normal_logic}
		A logic $\Lambda$ is \textbf{normal} if it satisfies:
		\begin{enumerate}
			\item for all $\varphi_1, \varphi_2 \in 
				\Phi_{\Theta}(L)$ and $\ell \in L$, the 
				formula $[\ell](\varphi_1 \rightarrow \varphi_2) \rightarrow 
				([\ell]\varphi_1 \rightarrow [\ell]\varphi_2 )$ is in 
				$\Lambda$;\footnote{Recall that $[\ell] \varphi = \neg \langle 
				\ell \rangle \neg \varphi$.}

			\item for all $\varphi \in \Lambda$ and $\ell \in L$, 
				$[\ell]\varphi \in \Lambda$.
		\end{enumerate}
		\end{definition}
		With this definition in hand, we can define the \emph{canonical model}.

		\begin{definition}[Canonical (Henkin) Model \cite{goldblatt_logics_1992,hollenberg_hennessy-milner_1995}]
		\label{def:canonical_model}
		Let $\Lambda \subseteq \Phi_{\Theta}(L)$ be a normal logic. 
		Then the \textbf{canonical model} is the Kripke structure 
		$\mathbf{C}^\Lambda = (S^\Lambda, \{R_\ell^\Lambda : \ell \in 
		L\}, V^\Lambda)$ defined as follows:
		\begin{itemize}
			\item $S^\Lambda$ is the set of states (worlds);

			\item for each $\ell \in L$, the transition relation 
				$R_\ell^\Lambda \subseteq S^\Lambda \times S^\Lambda$ is 
				defined such that $s R_\ell t$ if and only if $\varphi \in t$ 
				implies that $\langle \ell \rangle \varphi \in s$.

			\item the valuation $V^\Lambda : \Theta \rightarrow S^\Lambda$ 
				is defined such that $V(p) = \{ s \in S^\Lambda : p \in s\}$.
		\end{itemize}
		\end{definition}
		\begin{theorem}[The Canonical Model Satisfies the Henkin Property \cite{goldblatt_logics_1992,hollenberg_hennessy-milner_1995}]
		\label{thm:henkin_property_canonical_model}
		For any state $s$ in the canonical model $\mathbf{C}^\Lambda$ and any 
		formula $\varphi \in \Phi_{\Theta}(L)$:
		$ s \models \varphi \Longleftrightarrow \varphi \in s$
		(the aforementioned \textbf{Henkin property}).
		\end{theorem}
		%
		\subsubsection{Hennessy-Milner Classes for Kripke Structures} 
		\label{ssub:hennessy_milner_classes_for_kripke_structures}
		It is important to note that in Hennessy and Milner's definition of 
		image-finite STs, all subtrees must be image-finite: in other 
		words, the set of image-finite STs is \emph{closed under 
		subtrees}. Thus, one could think about generalizing Theorem 
		\ref{thm:hm_theorem_image_finite_sts} by examining when modal 
		equivalence implies bisimulation for a larger class of STs that 
		is closed under subtrees. The following definition captures that 
		spirit in the context of Kripke structures, but it does so without 
		insisting on image-finiteness.
		\begin{definition}[Visser/Hollenberg Hennessy-Milner Property \cite{hollenberg_hennessy-milner_1995}]
		\label{def:visser_hm_property}
		Let $\mathfrak{H}$ be a set of Kripke structures. $\mathfrak{H}$ is 
		said to satisfy the \textbf{Visser/Hollenberg Hennessy-Milner Property 
		(VHHM property) with respect to $\Phi_{\Theta}(L)$} if for any two 
		Kripke structures $\mathbf{S}, \mathbf{T} \in \mathfrak{H}$ and 
		\textbf{any two states} $s^\prime \in S$ and $t^\prime \in T$
			\begin{equation}
				s^\prime \; \bisim \; t^\prime \quad \Leftrightarrow 
				\quad s^\prime \approx_{\Phi_{\Theta}(L)} t^\prime.
			\end{equation}
		\end{definition}
		\begin{remark}
		\label{rem:hm_property}
		We use the terminology VHHM property to distinguish this property from  
		another definition of \textbf{Hennessy-Milner Property} in the 
		literature, which considers only modal equivalence and bisimulation 
		between initial states (the points of pointed Kripke structures). For 
		example, this definition is used in \cite{goranko_model_2007}. However, 
		we note that a number of other sources use what we call the VHHM 
		property; see \cite{goldblatt_saturation_1995} for example.
		\end{remark}
		\begin{definition}[Visser-Hollenberg Hennessy-Milner Class]
		\label{def:vh_hm_class}
		We say that any set of Kripke structures that satisfies the 
		Visser/Hollenberg Hennessy-Milner Property is a 
		\textbf{Visser-Hollenberg Hennessy-Milner class (VHHM class)}. 
		\end{definition}
		Definition \ref{def:visser_hm_property} seems innocuous, but in fact 
		the VHHM property is a nontrivial strengthening of the HM property 
		described in Remark \ref{rem:hm_property}. To the best of our 
		knowledge, there are no results in the literature that compare VHHM 
		classes with this alternate definition of HM classes. We will revisit 
		this in Section \ref{sub:properties_of_maximal_vhhm_classes_of_gsts} 
		where we exhibit a Kripke structure that fails to be a member of any 
		VHHM class because it fails to satisfy the conditions of Definition 
		\ref{def:visser_hm_property}.

		Importantly, though, there is an elegant characterization of 
		\emph{maximal} VHHM classes due to Visser and reported in 
		\cite{hollenberg_hennessy-milner_1995}. We 
		first define the notion of a ``Henkin-like'' model 
		\cite{hollenberg_hennessy-milner_1995}.

 		\begin{definition}[Henkin-like model \cite{hollenberg_hennessy-milner_1995}]
		\label{def:henkin_like_model}
		Let $\mathbf{C}^K$ be the canonical model associated with the smallest 
		normal logic $K$. Then a \textbf{Henkin-like} model is any Kripke 
		structure $\mathbf{HC}^K = (S^K, \{R^{\mathbf{H}K}_\ell \subseteq 
		R^K_\ell : \ell \in L\}, V^K)$ that satisfies the Henkin property (see 
		Theorem \ref{thm:henkin_property_canonical_model} and the discussion 
		preceding it).
		\end{definition}
		Thus, a Henkin-like model is simply the canonical model with  
		transitions removed in such a way that a state satisfies a formula if 
		and only that formula is an element of the state (recall that the 
		states in $\mathbf{C}^K$ are sets of formulas). Henkin-like models form 
		the basis for maximal VHHM classes in the following sense.
		\begin{theorem}[Maximal VHHM Classes \cite{hollenberg_hennessy-milner_1995}]
		\label{thm:vhhm_maximal_classes}
		Let $\mathbf{HC}^K$ be any Henkin-like model, and let 
		$\mathsf{S}(\mathbf{HC}^K)$ be the set of generated sub-models of 
		$\mathbf{HC}^K$. Then
		\begin{enumerate}
			\item  The set of all Kripke structures that are bisimilar to a 
				model in $\mathsf{S}(\mathbf{HC}^K)$ is a maximal VHHM class; 
				that is it is maximal in a set-theoretic sense. We denote such 
				a class by $\mathsf{BS}(\mathbf{HC}^K)$ .

			\item Let $\mathfrak{H}$ be any set of Kripke structures that 
				satisfies the VHHM property. Then $\mathfrak{H} \subseteq 
				\mathsf{BS}(\mathbf{HC}^K)$ for at least one Henkin-like model 
				$\mathbf{HC}^K$.
		\end{enumerate}
		\end{theorem}
		The basic idea behind Theorem \ref{thm:vhhm_maximal_classes} is this:  
		a set of models $\mathsf{BS}(\mathbf{HC}^K)$ is necessarily a VHHM 
		class because modal equivalence is a bisimulation relation over a 
		single Henkin-like model (each maximal set of formulas is satisfied 
		only by its own unique state in the model). Thus, Henkin-like models 
		effectively ``canonicalize'' different VHHM classes because a given 
		Henkin-like model associates a particular transition structure with 
		each and every (maximal) set of formulas that can be satisfied in any 
		Kripke structure ($K$ is sound and complete with respect to Kripke 
		structures).

		Maximal VHHM classes are related to the VHHM class of image finite 
		models in the following way.
		\begin{theorem}[Image Finite Kripke structures \cite{hollenberg_hennessy-milner_1995}]
			Each maximal VHHM class of Theorem \ref{thm:vhhm_maximal_classes} 
			contains every Kripke structure that is bisimilar to an image 
			finite Kripke structure. Hence, each maximal VHHM class contains 
			all image finite Kripke structures, and the class of image finite 
			Kripke structures is itself a VHHM class.
		\end{theorem}

\section{Generalized Hennessy-Milner Logic} 
\label{sec:generalized_hennessy_milner_logic}
	In this section, our aim is to define a logic akin to HML but with GSTs as 
	the intended models. We proceed by first defining the syntax and then the 
	semantics of our logic.
	\subsection{HML for GSTs: Syntax} 
	\label{sub:hml_for_gsts_syntax}
	Our generalization of HML will be mostly recognizable, but the $\langle 
	\thickspace \rangle$ modality requires some significant modifications. In 
	particular, recall that in weak bisimulation for GSTs, transitions are 
	replaced by \emph{trajectories} (see Definition \ref{def:trajectory}) and 
	labels by functions over trajectories. To capture this notion, we
	generalize the way we \emph{label} the $\langle \; \rangle$ modality. 
	\begin{definition}[Domain of modalities]
	\label{def:domain_of_modalities}
	A domain of modalities is a totally ordered set, $(\mathcal{I}, 
	\preceq_\mathcal{I})$, together with a set of labels $L$.
	\end{definition}
	Intuitively, a domain of modalities will be used to define the 
	trajectory-like structures appearing in our modalities. However, we 
	eventually need such a domain of modalities to satisfy some additional 
	properties to ensure certain formulas exist. Hence, we provide the 
	following definitions. 

	\begin{definition}[Spanned by an interval]
		\label{def:span}
		Given a totally ordered set $\mathcal{I}$, we say a subset $I \subseteq 
		\mathcal{I}$ is \textbf{spanned by an interval}, if there exists an 
		interval $[i_0,i_1] = \{i \in \mathcal{I} : i_0 \preceq_\mathcal{I} i 
		\preceq_\mathcal{I} i_1\}$ such that $I \subseteq [i_0, i_1]$ and 
		$\{i_0,i_1\} \subset I$; this is equivalent to saying that $I$ contains 
		its least upper bound (LUB) and greatest lower bound (GLB). We say that 
		$i_0$ and $i_1$ are the \textbf{left and right endpoints of} $I$, 
		respectively, and they will be denoted by $\leftendpoint{I}$ and 
		$\rightendpoint{I}$, respectively. 
	\end{definition}
	\begin{definition}[Left-open subset]
		\label{def:left_open}
		A subset $I$ of a totally ordered set $\mathcal{I}$ is \textbf{left open} 
		if there exists a set $I^\prime \subseteq \mathcal{I}$ spanned by an 
		interval such that $I = I^\prime \backslash \{\leftendpoint{I^\prime}\}$
	\end{definition}
	\begin{definition}[Closed under left-open concatenation]
		\label{def:closed_under_spans}
		We say that a totally ordered set $\mathcal{I}$ is \textbf{closed under 
		left-open concatenation} if for any two left-open subsets $I_1, I_2 
		\subseteq \mathcal{I}$, there exists another left-open set $I_3$ such 
		that there is an order preserving bijection from $I_3$ to the totally 
		ordered set $I_1 ; I_2 = (\{1\} \times I_1 ) \cup ( \{2\} \times I_2 )$ 
		under the lexicographic ordering. A totally ordered set $\mathcal{I}$ 
		that is closed under left-open concatenation will be denoted 
		$\bar{\mathcal{I}}$.
	\end{definition}
	\begin{example}
	Any totally ordered set that can be embedded in an order-preserving 
	additive group structure is closed under left-open concatenation. 
	$\mathbb{N}$, $\mathbb{R}$ and $\mathbb{R} \times \mathbb{N}$ are examples.
	\end{example}
	\begin{remark}
		Henceforth, we will work exclusively with total orders that are closed 
		under left-open concatenation when we construct a domain of modalities.
	\end{remark}
	\begin{definition}[Modal execution]
		Let $(\bar{\mathcal{I}},L)$ be a domain of modalities. A \textbf{modal 
		execution} is a map from a left-open subset of $\bar{\mathcal{I}}$ to 
		the set of labels, $L$. The set of modal executions over 
		$(\bar{\mathcal{I}},L)$ will be denoted 
		$\mathcal{M}(\bar{\mathcal{I}},L)$.
	\end{definition}
	The notion of a modal execution is almost usable as a label for our 
	generalized diamond modalities, but it is too tied to the specific 
	\emph{domain} of the function in question. This will prove cumbersome in 
	the future, so we restrict ourselves to equivalence classes of such 
	modalities.

	\begin{definition}[Order Equivalent Modal Executions]
		\label{def:order_equiv_modal_executions}
		Let $E_1: I_1 \rightarrow L$ and $E_2 : I_2 \rightarrow L$ be two modal 
		executions from a domain of modalities $(\bar{\mathcal{I}},L)$. We say 
		that $E_1$ is \textbf{order equivalent} to $E_2$ if there exists an order 
		preserving bijection $\lambda : I_1 \rightarrow I_2$ such that $E_1(i) 
		= E_2(\lambda(i))$ for all $i \in I_1$. If $E_1$ is order equivalent to 
		$E_2$, then we write $E_1 \overset{\text{o.e}}{\sim} E_2$. This 
		definition parallels Definition \ref{def:order-equivalence} for GST 
		trajectories.
	\end{definition}
	\begin{theorem}[Order Equivalence is an equivalence relation]
		\label{thm:order_equiv_class}
		$\overset{\text{o.e.}}{\sim}$ is an equivalence relation between modal 
		executions. We denote by the equivalence class $\{ E^\prime \in 
		\mathcal{M}(\bar{\mathcal{I}},L) : E \overset{\text{o.e.}}{\sim} 
		E^\prime \}$ by $|E|$, and the set of such equivalence classes by 
		$|\mathcal{M}(\bar{\mathcal{I}},L)|$. 
	\end{theorem}
	\begin{definition}[Set of Generalized HML (GHML) formulas]
	\label{def:ghml_formulas}
	Given a domain of modalities $(\bar{\mathcal{I}},L)$, the set of 
	\textbf{Generalized HML (GHML)} formulas is the set of formulas, 
	$\Phi_{\text{GHML}}(\bar{\mathcal{I}},L)$, inductively defined according to 
	the following rules: 
		\begin{equation}
			\varphi :=  \thickspace \top \quad | 
				\quad \neg \varphi \quad | 
				 \quad \varphi_1 \wedge \varphi_2 \quad | 
				 \quad \ghmldiamond{|E|} \varphi
		\end{equation}
		where $|E|$ is an equivalence class of modal executions over the domain 
		of modalities $(\bar{\mathcal{I}},L)$.
	\end{definition}
	The formal semantics of this logic will be presented in next subsection 
	within Definition \ref{def:ghml_semantics}.

	We have chosen to define our logic without propositional variables in order 
	to mirror Hennessy and Milner's original work. However, in Section 
	\ref{sec:maximal_hennessy_milner_classes_for_gsts} we will consider a modal 
	logic with a syntax based on Definition \ref{def:ghml_formulas}, and so we 
	describe here such a modal logic.
	\begin{definition}[GHML Modal Logic]
		\label{def:ghml_modal_logic}
		A GHML modal logic is a modal logic with all of the connectives from 
		Definition \ref{def:ghml_formulas} plus propositional variables. If 
		$\Theta$ is the set of propositional variables, then we denote the set 
		of these formulas by 
		$\Phi_{\text{GHML-}\Theta}(\bar{\mathcal{I}},L)$.
	\end{definition}
	A number of the proof theoretic results from Section 
	\ref{ssub:the_canonical_model_for_a_modal_logic} apply equally well to a 
	GHML modal logic: the definition of a logic (Definition \ref{def:logic}), 
	the definition of $\Lambda$-consistency and the definition of 
	$\Lambda$-maximality all apply directly to a GHML modal logic. On the other 
	hand, Lindenbaum's lemma (Lemma \ref{lemma:lindenbaums_lemma}) requires a 
	different proof because of the multiplicity of modalities. Nevertheless, it 
	is still true, as the following theorem asserts.
	\begin{theorem}[$\Lambda$-maximal sets of GHML formulas]
		\label{thm:maximal_consistent_ghml_sets}
		Let $\Lambda \subseteq 
		\Phi_{\text{GHML-}\Theta}(\bar{\mathcal{I}},L)$ be a logic, and let 
		$\Gamma \subseteq \Phi_{\text{GHML-}\Theta}(\bar{\mathcal{I}},L)$ 
		be a $\Lambda$-consistent set of formulas. Then there exists a 
		$\Lambda$-maximal set $\Gamma_0 \subseteq 
		\Phi_{\text{GHML-}\Theta}(\bar{\mathcal{I}},L)$ such that $\Gamma 
		\subseteq \Gamma_0$.
	\end{theorem}
	\begin{proof}
		Because the collection of GHML modal logic formulas is a set, this is a 
		straightforward application of Zorn's lemma.
	\end{proof}

	\subsection{HML for GSTs: Semantics} 
	\label{sub:hml_for_gsts_semantics}
	We define the semantics of GHML for a GST model $G$ in terms of the set of 
	the sub-GSTs of $G$; because each GST is itself defined in 
	terms of sets, we may soundly define the following notion of a sub-GST 
	rooted at a node.
	\begin{definition}[Sub-GST rooted at a node]
		Let $G = (P, \preceq, p_0, \mathcal{L})$ be a GST. We let $G |_{p}$ 
		denote the sub-GST of $G$ rooted at $p$, i.e. $G|_p \triangleq 
		(\{p^\prime \in P | p^\prime \succeq p \}, \preceq, p, \mathcal{L})$.
	\end{definition}
	Now we can formally define the semantics of the generalized HML formulas 
	defined above.
	\begin{definition}[Satisfaction relation over GHML formulas]
	\label{def:ghml_semantics}
		Let $G = (P,  \preceq, p_0, \mathcal{L})$ be a GST, and let 
		$\mathcal{G}_{\text{sub}} := \{ G|_p : p \in P \}$. A satisfaction 
		relation, $\models$, is a relation $\models \thickspace \subseteq 
		\thickspace \mathcal{G}_{\text{sub}} \times 
		\Phi_{\text{GHML}}(\bar{\mathcal{I}},L)$ that is defined inductively 
		over GHML formulas.  Satisfaction of the formula $\ghmldiamond{|E|} 
		\varphi$ is defined in the following way: $G \models 
		\ghmldiamond{|E|}\varphi$ if and only if there exists an interval 
		$(p_0, p]$; a left-open set $I \subset \bar{\mathcal{I}}$; and an order 
		preserving bijection $\lambda : I \rightarrow (p_0, p]$ such that 
		\begin{enumerate}
			\item $\mathcal{L} \circ \lambda \in |E|$

			\item $G|_p \models \varphi$.
		\end{enumerate}
		The satisfaction relation is defined for other formulas in the usual 
		way.
	\end{definition}

	Intuitively, a GST satisfies the formula $\ghmldiamond{|E|} \top$ when it 
	has a trajectory emanating from its root that is order equivalent to every 
	$E \in |E| $ (recall that all elements of $|E|$  are order equivalent to 
	each other). Importantly, this logic also yields formulas that are 
	analogous to HML formulas on discrete GSTs when there are at least two 
	points in $\bar{\mathcal{I}}$. In particular, if $i_0 \preceq i_1$, then 
	$\{i_0, i_1\}$ is spanned by the interval $[i_0,i_1]$, and the singleton 
	point $\{i_1\}$ is a left-open set. Thus, 
	$\mathcal{M}(\bar{\mathcal{I}},L)$ contains modal executions that are 
	order-equivalent to discrete transitions in a GST. Of course, discrete 
	transitions are the essence of the semantics for the labeled modalities in 
	HML. 

\section{A First Hennessy-Milner Theorem: ``Image-finite'' GSTs} 
\label{sec:weak_bisimulation_preserves_ghml_modal_equivalence}
	In this section, our objective is to define something like a class of 
	image-finite GSTs with the ultimate intention of defining a Hennessy-Milner 
	class of GSTs. We introduce this section with an example to show that the 
	most straightforward definition of image-finiteness is too exclusive to be 
	of much interest.
	\begin{example}
		\label{ex:strict_image_finite_gsts}
		Consider the following GST defined on the unit interval $[0,1] \subset 
		\mathbb{R}$: $G_{[0,1]} := ([0,1],\leq_\mathbb{R},0,(0,1] \rightarrow 
		\{\alpha\})$.
	\end{example}
	The point of Example \ref{ex:strict_image_finite_gsts} is that $G_{[0,1]}$ 
	has uncountably many nodes that are accessible from the root, $0$, with a 
	single trajectory: that is for any $x,y \in (0,1]$ there is an order 
	preserving bijection between $(0,x]$ and $(0,y]$. Since these trajectories' 
	nodes are labeled by a single label, $\alpha$, they are thus order 
	equivalent in the sense of Definition \ref{def:order-equivalence}. 
	Nevertheless, this GST appears to be about as simple as one could wish for 
	in terms of nondeterminism: there is essentially no branching behavior at 
	all.

	\subsection{GSTs as Discrete Structures} 
	\label{sub:gsts_as_discrete_structures}
	The discussion following Example \ref{ex:strict_image_finite_gsts} suggests 
	a way of looking at GSTs that will be profitable, especially when it comes 
	to examining GHML formulas and constructing Hennessy-Milner classes. In 
	particular, we use equivalence classes of modal executions to label 
	discrete transitions on a Kripke structure; we show that such a 
	construction captures the relevant structure of a given class of GSTs with 
	respect to bisimulation and GHML satisfaction.
	\begin{definition}[Captured by a Domain of modalities]
	\label{def:captured}
		Let $\mathcal{U}$ be a set of GSTs. We say that $\mathcal{U}$ is 
		\textbf{captured} by a domain of modalities $(\bar{\mathcal{I}},L)$ if 
		every trajectory from every GST in $\mathcal{U}$ is order equivalent to 
		some modal execution over $(\bar{\mathcal{I}},L)$. 
	\end{definition} 
	\begin{definition}[Surrogate Kripke Structure]
		\label{def:surrogate_lts}
		Let $\mathcal{U}$ be a set of GSTs that is captured by a domain of 
		modalities $(\bar{\mathcal{I}},L)$. For any GST $G = 
		(P,\preceq_P,p_0,\mathcal{L})$ in $\mathcal{U}$, we define a 
		\textbf{surrogate Kripke structure}, $\mathbf{G} = (P, \{R^G_{|E|} 
		\subseteq P \times P : |E| \in \mathcal{M}(\bar{\mathcal{I}},L) \}, 
		V)$, as follows:
		\begin{itemize}
			\item the set of states is $P$; and 

			\item $p_1 \overset{|E|}{\rightarrow} p_2$ -- i.e. $p_1 
				R^G_{|E|} p_2$ -- if and only if $p_1 \preceq_P p_2$ and 
				$(p_1,p_2]$ is order equivalent to an element of $|E|$; and

			\item $V : \Theta \rightarrow \{ P \}$ indicates all 
				propositional variables are true in all states in $P$.
		\end{itemize}
	\end{definition}
	\begin{remark}
		We will not consider valuations in this section, but they will be used 
		in the next section. Thus, for the purposes of this section, we may 
		regard surrogate Kripke structures as labeled transition systems.
	\end{remark}
	\begin{example}[Surrogate Kripke Structure for $G_{[0,1]}$]
		\label{ex:first_surrogate_ks}
		If we let $\bar{\mathcal{I}} = (\mathbb{R},\leq_\mathbb{R})$ and $L = 
		\{\alpha\}$, then Figure \ref{fig:surrogate_ks} shows some of the 
		transitions that appear in the surrogate Kripke structure for the GST 
		$G_{[0,1]}$ from Example \ref{ex:strict_image_finite_gsts}.
	\end{example}
	\begin{figure}
			\centering
			\includegraphics[width=0.5\textwidth]{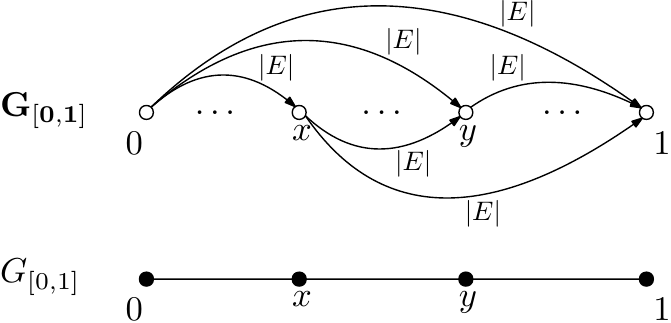}
			\caption{GST and surrogate Kripke structure from Example \ref{ex:first_surrogate_ks}; we define $E : (0,1] \rightarrow \{\alpha\}$.}
			\label{fig:surrogate_ks}
	\end{figure}

	The idea of a surrogate Kripke structure seems simple enough, but its 
	importance is indicated by the following two theorems: one relates GHML 
	formulas to HML formulas, and the other relates ordinary bisimulation to 
	weak bisimulation for GSTs.

	\begin{theorem}[Relating GHML formulas on $G$ to HML formulas on $\mathbf{G}$]
		\label{thm:ghml_hml_surrogates}
		Let $(\bar{\mathcal{I}},L)$ be a domain of modalities, and let 
		$\mathcal{U}$ be a set of GSTs captured by $(\bar{\mathcal{I}},L)$. 
		Furthermore, consider  HML over the set of labels given by 
		$|\mathcal{M}(\bar{\mathcal{I}},L)|$. Then for every $G = (P, 
		\preceq_P, p_0,\mathcal{L}) \in \mathcal{U}$, 
		\begin{enumerate}
			\item for all $\varphi \in \Phi_{\text{GHML}}(\bar{\mathcal{I}},L)$, 
				$G \models \varphi ~ \Rightarrow ~ p_0 \models \varphi_{\langle 
				\rangle }$ and

			\item for all $\phi \in 
				\Phi_{\text{HML}}(|\mathcal{M}(\bar{\mathcal{I}},L)|)$, $p_0 
				\models \phi ~ \Rightarrow ~ G \models \phi_{\ghmldiamond{}}$. 
		\end{enumerate}
		The notation $\varphi_{\langle\rangle}$ indicates that the GHML formula 
		$\varphi$ is converted to an HML formula by replacing each 
		$\ghmldiamond{|E|}$ modality with the corresponding HML modality 
		$\langle |E| \rangle$. $\phi_{\ghmldiamond{}}$ indicates an analogous 
		conversion from an HML formula to a GHML formula.
	\end{theorem}
	\begin{proof}
		This is a straightforward proof by induction on formula structure (the 
		base case is $\top$, which has an identical meaning in HML and GHML). 
		The ability to match HML modalities to GHML modalities (and conversely) 
		is assured by the way we have constructed the surrogate Kripke 
		structure, and in particular, the fact that we have labeled 
		trajectories by \emph{equivalence classes} of modal executions.
	\end{proof}
	\begin{theorem}[Weak bisimulation between GSTs and bisimulation between surrogates]
		\label{thm:bisim_weak_bisim_surrogates}
		Let $\mathcal{U}$ and $(\bar{\mathcal{I}},L)$ be as in Theorem 
		\ref{thm:ghml_hml_surrogates}. Furthermore, let $G_1 = 
		(P,\preceq_P,p_0,\mathcal{L}_P)$ and $G_2 = (Q, \preceq_Q, q_0, 
		\mathcal{L}_Q)$ be two GSTs in $\mathcal{U}$. Then
		\begin{equation}
			G_1 \; \bisim_w \; G_2 \quad \Longleftrightarrow \quad p_0 \; \bisim \; q_0.
		\end{equation}
		where the bisimulation $p_0 \; \bisim \; q_0$ is taken in the context 
		of the surrogate Kripke structures $\mathbf{G_1}$ and $\mathbf{G_2}$.
	\end{theorem}
	\begin{proof}
		This theorem, like Theorem \ref{thm:ghml_hml_surrogates}, is a 
		consequence of the way that we defined the surrogate Kripke structure: 
		in particular, any weak bisimulation relation between $G_1$ and $G_2$ 
		is a bisimulation relation between $\mathbf{G_1}$ and $\mathbf{G_2}$ 
		and conversely.
	\end{proof}
	Theorems \ref{thm:ghml_hml_surrogates} and 
	\ref{thm:bisim_weak_bisim_surrogates} together reinforce that weak 
	bisimulation is very much a discrete notion. In the context of GHML 
	formulas and the construction of Hennessy-Milner classes, though, this will 
	prove to be an advantage.

	\subsection{``Image-Finite'' GSTs} 
	\label{sub:image_finite_gsts}
	If we reconsider Example \ref{ex:strict_image_finite_gsts} in the context 
	of Theorems \ref{thm:ghml_hml_surrogates} and 
	\ref{thm:bisim_weak_bisim_surrogates}, then a natural means of defining 
	``image-finite'' GSTs emerges. In particular, it is evident that the 
	surrogate Kripke structure $\mathbf{G}_{[0,1]}$ (see Figure 
	\ref{fig:surrogate_ks}) is bisimilar to the two state Kripke structure 
	$\mathbf{M}$ depicted in Figure \ref{fig:image_finite_gst_example}.
	\begin{figure}
			\centering
			\hspace*{-1.5cm}\includegraphics[width=0.3\textwidth]{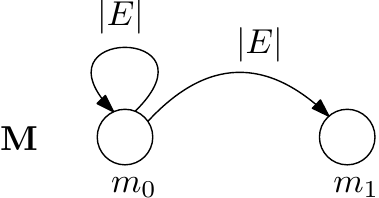}
			\caption{A Kripke structure $\mathbf{M}$ that is bisimilar to $\mathbf{G_{[0,1]}}$; again, we define $E : (0,1] \rightarrow \{\alpha\}$.}
			\label{fig:image_finite_gst_example}
	\end{figure}
	Of course $\mathbf{M}$ \emph{is} image-finite, and so we have just 
	exemplified a serviceable means by which we can define image-finiteness for 
	GSTs.
	\begin{definition}[Image-finite GST]
		\label{def:image_finite_gst}
		Let $\mathcal{U}$ and $(\bar{\mathcal{I}},L)$ be as before. Then a GST 
		$G \in \mathcal{U}$ is \textbf{image finite} if its surrogate 
		Kripke structure is $\bisim_K$ to an image-finite Kripke 
		structure.
	\end{definition}
		In Definition \ref{def:image_finite_gst}, we really mean bisimulation 
		according to $\bisim_K$ of Definition 
		\ref{def:bisim_kripke_structures}, so that the surrogate Kripke 
		structure ends up being bisimilar to a Kripke structure that also has a 
		``universal'' valuation. Hence, the surrogate Kripke structure and its 
		image-finite pair can be regarded simply as labeled transition systems 
		with the usual notion of bisimulation $\bisim$. We introduce this 
		requirement in preparation for the treatment of maximal classes to come.

	Of course because of Theorem \ref{thm:ghml_hml_surrogates} and 
	\ref{thm:bisim_weak_bisim_surrogates}, this definition of image-finiteness 
	implies a Hennessy-Milner class of GSTs through the use of Hennessy and 
	Milner's original theorem.
	\begin{theorem}[Image-finite GSTs form a Hennessy-Milner class]
		\label{thm:image_finite_gst_hm_class}
		Let $\mathcal{U}$ and $(\bar{\mathcal{I}},L)$ be as before. Then the 
		set of image-finite GSTs in $\mathcal{U}$ forms a Hennessy-Milner class 
		according to weak bisimulation. That is any two image finite GSTs from 
		$\mathcal{U}$ are weakly bisimilar if and only if they satisfy the same 
		GHML formulas.
	\end{theorem}

\section{Maximal Hennessy-Milner Classes for GSTs} 
\label{sec:maximal_hennessy_milner_classes_for_gsts}
	The construction of surrogate Kripke structures in Definition 
	\ref{def:surrogate_lts} combined with Theorems 
	\ref{thm:ghml_hml_surrogates} and \ref{thm:bisim_weak_bisim_surrogates} 
	suggests that the maximal VHHM classes of Section 
	\ref{ssub:hennessy_milner_classes_for_kripke_structures} have analogs as 
	maximal HM classes of GSTs with respect to weak bisimulation. In this 
	section we demonstrate that this is indeed the case, although the 
	translation is not exact. We also exhibit some interesting GST-specific 
	properties that these classes possess.

	\subsection{Characterizing Maximal VHHM Classes of GSTs} 
	\label{sub:maximal_vhhm_classes_of_gsts}
	The essential assumption required for the proof of Theorem 
	\ref{thm:vhhm_maximal_classes} is the VHHM property: that is a VHHM class 
	of Kripke structures cannot contain two states that satisfy the same 
	formulas yet are not bisimilar. Since we are interested in weak 
	bisimulation and GHML formulas, we can straightforwardly define a VHHM 
	property for GSTs as follows.
	\begin{definition}[VHHM class for GSTs]
		\label{def:vhhm_class_gsts}
		Let $\mathcal{U}$ be a set of GSTs, and let $(\bar{\mathcal{I}},L)$ be 
		a domain of modalities that captures $\mathcal{U}$, so that 
		$\approx_{\text{GHML}}$ is interpreted with respect to 
		$\Phi_{\text{GHML}-\Theta}(\bar{\mathcal{I}},L)$. Then we say that a 
		subset $\mathfrak{h} \subseteq \mathcal{U}$ \textbf{satisfies the VHHM 
		property for GSTs} if for any two sub-GSTs $G_1|_p$ and $G_2|_q$ from 
		the set $\mathfrak{h}$ (possibly with $G_1 = G_2$),
		\begin{equation}
			G_1|_p \;\; \bisim_w \; G_2|_q \quad \Longleftrightarrow \quad G_1|_p \; \approx_{\text{GHML}} \; G_2|_q.
		\end{equation}
	\end{definition}

	Of course this definition will help us define maximal VHHM classes of GSTs 
	because of Theorem \ref{thm:ghml_hml_surrogates} and 
	\ref{thm:bisim_weak_bisim_surrogates}, which relate GHML formulas and weak 
	bisimulation for GSTs to HML formulas and bisimulation for Kripke 
	structures. Hence, we have the following theorem.
	\begin{theorem}[Maximal VHHM classes for GSTs are restrained by maximal VHHM classes for their surrogates]
		\label{thm:gst_vhhm_property_to_kripke_vhhm_property}
		Let $\mathcal{U}$ and $(\bar{\mathcal{I}},L)$ be as in Definition 
		\ref{def:vhhm_class_gsts}. If $\mathfrak{h} \subseteq \mathcal{U}$ is a 
		VHHM class of GSTs, then the set of surrogate Kripke structures 
		$\{\mathbf{G}: G \in \mathfrak{h}\}$ satisfies the VHHM property of 
		Definition \ref{def:vh_hm_class} with respect to 
		$\Phi_{\text{HML}}(|\mathcal{M}(\bar{\mathcal{I}},L)|)$.
	\end{theorem}
	\begin{proof}
		This is a direct consequence of Theorems \ref{thm:ghml_hml_surrogates} 
		and \ref{thm:bisim_weak_bisim_surrogates}. First, check that for any 
		node $p$ in a surrogate Kripke structure $\mathbf{G_1}$ and node $q$ in 
		surrogate Kripke structure $\mathbf{G_2}$, $p \; \approx_{\text{HML}} 
		\; q$ implies $p \; \bisim_K \; q$. Because of Theorem 
		\ref{thm:ghml_hml_surrogates}, we know that $p \; \approx_{\text{HML}} 
		\; q$ implies $G_1|_p \approx_{\text{GHML}} G_2|_q$. But $\mathfrak{h}$ 
		is a VHMM class of GSTs, so the preceding implies that $G_1|_p \; 
		\bisim_{w} \; G_2|_q$, and Theorem 
		\ref{thm:bisim_weak_bisim_surrogates} then implies that $p \; \bisim_K 
		\; q$ as required. The converse follows by using first Theorem 
		\ref{thm:bisim_weak_bisim_surrogates} and then Theorem 
		\ref{thm:ghml_hml_surrogates}.
	\end{proof}

	The essential intuition here is that \emph{any} set of GSTs that satisfies 
	the VHHM property will yield a set of surrogate Kripke structures that 
	satisfies the VHHM property; then by part 2 of Theorem 
	\ref{thm:vhhm_maximal_classes}, these surrogate Kripke structures will be 
	contained in a maximal VHHM class of Kripke structures. Thus, a VHHM class 
	of GSTs can only be enlarged so long as its surrogate Kripke structures 
	do not escape a maximal VHHM class of Kripke structures, so \emph{every 
	maximal VHHM class of GSTs can be matched to at least one maximal VHHM 
	class of Kripke structures}. This is expressed in the following corollary.
	\begin{corollary}
	\label{cor:vhhml_class_bound_normal_logic}
		Let $\mathcal{U}$ and $(\bar{\mathcal{I}},L)$ be as in Definition 
		\ref{def:vhhm_class_gsts}. If $\mathfrak{h} \subseteq \mathcal{U}$ is a 
		VHHM class of GSTs, then there exists a Henkin-like model 
		$\mathbf{HC}^K$ such that $\{\mathbf{G} : G \in \mathfrak{h}\} 
		\subseteq \mathsf{BS}(\mathbf{HC}^K)$. Furthermore, if there is a set 
		$\mathfrak{h}^\prime \subseteq \mathcal{U}$ such that $\mathfrak{h} 
		\subseteq \mathfrak{h}^\prime$ and $\{\mathbf{G} : G \in 
		\mathfrak{h}^\prime\}  \subseteq \mathsf{BS}(\mathbf{HC}^K)$, then 
		$\mathfrak{h}^\prime$ is a VHHM class of GSTs.
	\end{corollary}

	However, we have not yet established that every maximal VHHM class of 
	Kripke structures corresponds to a maximal VHHM class of GSTs. Indeed, the 
	fact that Theorem \ref{thm:vhhm_maximal_classes} makes no assumptions about 
	valuations immediately suggests that several maximal VHHM classes of the 
	form $\mathsf{BS}(\mathbf{HC}^K)$ will correspond to the same maximal VHHM 
	class of GSTs. As it turns out, there are other yet more profound 
	redundancies in the Henkin-like models derived from the canonical model 
	over the smallest normal logic, $K$. These differences are described in the 
	following theorem, though it too falls short of an absolute 
	characterization of maximal VHHM classes for GSTs.
	\begin{theorem}[Maximal VHHM classes for GSTs and refined Henkin-like models]
		\label{thm:smaller_vhhm_classes_gsts}
		Let $\mathcal{U}$ and $(\bar{\mathcal{I}},L)$ be as in Definition 
		\ref{def:vhhm_class_gsts}, and let $\mathfrak{h} \subseteq \mathcal{U}$ 
		be a VHHM class of GSTs. Furthermore, let $\Delta$ be the smallest 
		normal logic that contains all of the following:
		\begin{itemize}
			\item the propositional variables $\Theta$;

			\item $ \forall |E_1|, |E_2| \in 
				|\mathcal{M}(\bar{\mathcal{I}},L)|$, the schema $\langle |E_1| 
				\rangle \langle |E_2| \rangle \varphi \rightarrow \langle 
				|E_{1;2}| \rangle \varphi$ ; and 

			\item $\forall |E|, |E_1|, |E_2| \in 
				|\mathcal{M}(\bar{\mathcal{I}},L)|$ such that there is an order 
				equivalence $\lambda : \; I_1;I_2 \rightarrow \text{dom}(E)$ 
				with $E\circ\lambda(1,\cdot) \in |E_1|$ and 
				$E\circ\lambda(2,\cdot) \in |E_2|$, the schema $\langle |E| 
				\rangle \varphi \rightarrow \langle |E_1| \rangle \langle |E_2| 
				\rangle  \varphi$.
		\end{itemize}
		Then $\{\mathbf{G} : G \in \mathfrak{h}\} \subseteq 
		\mathsf{BS}(\mathbf{HC}^\Delta)$ for some Henkin-like model 
		$\mathbf{HC}^\Delta$ that preserves the first-order transition-relation 
		properties imposed on $\mathbf{C}^\Delta$ by the schemata above. 
		Furthermore, if there is a set $\mathfrak{h}^\prime \subseteq 
		\mathcal{U}$ such that $\mathfrak{h} \subseteq \mathfrak{h}^\prime$ and 
		$\{\mathbf{G} : G \in \mathfrak{h}^\prime\} \subseteq 
		\mathsf{BS}(\mathbf{HC}^\Delta)$, then $\mathfrak{h}^\prime$ is a VHHM 
		class of GSTs.
	\end{theorem}
	\begin{proof}(Theorem \ref{thm:smaller_vhhm_classes_gsts}.)
		All of the additions to the logic $\Delta$ reflect structure in 
		surrogate Kripke structures (including in the valuations used), and 
		Theorem \ref{thm:vhhm_maximal_classes} remains applicable when confined 
		to such restricted Kripke structures. 
	\end{proof}
		The reader will recognize in the formulas $\langle |E_1| \rangle 
		\langle |E_2| \rangle \varphi \rightarrow \langle |E_{1;2}| \rangle 
		\varphi$ and $\langle |E| \rangle \varphi \rightarrow \langle |E_1| 
		\rangle \langle |E_2| \rangle  \varphi$ the schemata for something like 
		\emph{transitivity} and \emph{weak density}, respectively 
		\cite{goldblatt_logics_1992}. That the surrogate Kripke structures 
		satisfy these conditions is a reflection of the unique semantics we 
		have specified for GHML: in particular, following one trajectory in a 
		GST followed by another implies the existence of a third, ``longer'' 
		trajectory (transitivity), and following a non-trivial trajectory 
		implies the existence of ``smaller'' trajectories ending and beginning 
		from some intermediary point (weak density).
		On the other hand, the additional constraint on the Henkin-like model 
		in Theorem \ref{thm:smaller_vhhm_classes_gsts} is necessary because 
		just satisfying the relevant schemata under one valuation is not enough 
		to impose the first-order transition relation properties that surrogate 
		Kripke structures possess (see \cite{goldblatt_logics_1992}).

		It is also worth noting that our choice of \emph{equivalence classes} 
		of modal executions is relevant here. Had we not chosen to label 
		transitions in the surrogate Kripke structure with such 
		\emph{equivalence classes}, there would be \emph{multiple} 
		order-equivalent transitions between any two nodes. This would lead to 
		additional Henkin-like models that fail to respect the semantics of 
		weak bisimulation: i.e. among a collection of order-equivalent 
		transitions, some could be present in the Henkin-like model while some 
		could be absent.

		Finally, it is important to note that neither Corollary 
		\ref{cor:vhhml_class_bound_normal_logic} nor Theorem 
		\ref{thm:smaller_vhhm_classes_gsts} imply that \emph{every} maximal 
		VHHM class of Kripke structures corresponds to a maximal VHHM class of 
		GSTs in $\mathcal{U}$. For one, the set $\mathcal{U}$ may be deficient. 
		For another, it remains as future work to show that every Henkin-like 
		model over the canonical model for logic $\Delta$ reflects the 
		surrogate Kripke structure of some GST.
	\subsection{Properties of Maximal VHHM Classes of GSTs} 
	\label{sub:properties_of_maximal_vhhm_classes_of_gsts}
	In this subsection we make two small remarks that identify some properties 
	of interest with regard to maximal VHHM classes of GSTs.

	First, we note that maximal VHHM classes are not so small that modal 
	equivalence within such a class implies \emph{strongly bisimulation} 
	(Definition \ref{def:strong-simulation}). That is to say there is a maximal 
	VHHM class which contains two GSTs that satisfy the same formulas yet are 
	not strongly bisimilar. Such a situation is illustrated in the following 
	example.
	\begin{example}
		\label{ex:modally_equiv_gsts_not_strongly_bisimilar}
		Consider the domain of modalities given by the set $\bar{\mathcal{I}} = 
		(\mathbb{R},\leq_\mathbb{R})$ and the set $L = \{\alpha,\beta\}$. 
		Furthermore, for a subset $A$ of $[0,1] \subset \mathbb{R}$, define the 
		GST $G_A$ as $G_A = ([0,1] \cup \left(\{1\} \times A 
		\right),\preceq_A,0,\mathcal{L}_A)$ where $\preceq_A = \cup_{a\in A} \{ 
		(x, (1,a)) : x \leq a\} ~\cup \leq_{\mathbb{R} \cap [0,1]}$ and 
		$\mathcal{L}_A : x \mapsto \alpha ~;  (1,a) \mapsto \beta$. We claim 
		that the GSTs $G_{\mathbb{Q} \cap (0,1)}$ and $G_{(0,1)\backslash 
		\mathbb{Q}}$ together satisfy the VHHM property: in fact their 
		surrogate Kripke structures are both bisimilar to the same image-finite 
		Kripke structure. Nevertheless, they are clearly not strongly 
		bisimilar, since there is no order preserving way of matching 
		$\mathbb{Q}\cap (0,1)$ with $(0,1)\backslash \mathbb{Q}$.
	\end{example}

	Second, we note that there are GSTs that don't belong to any VHHM class. 
	This is ultimately because there are Kripke structures that don't belong to 
	any VHHM class of Kripke structures: the following example describes just 
	such a Kripke structure.
	\begin{example}
	\label{ex:non_self_bisimilar_ks}
		Consider the Kripke structure depicted in Figure 
		\ref{fig:non_self_bisimilar_ks} with a valuation that assigns all 
		propositional variables to be true in all states. We claim that the 
		shaded states satisfy the same formulas, yet they are clearly not 
		bisimilar. Hence, this Kripke structure doesn't belong to any VHHM 
		class of Kripke structures.
	\end{example}
	\begin{figure}
		\centering
		\includegraphics[width=0.9\textwidth]{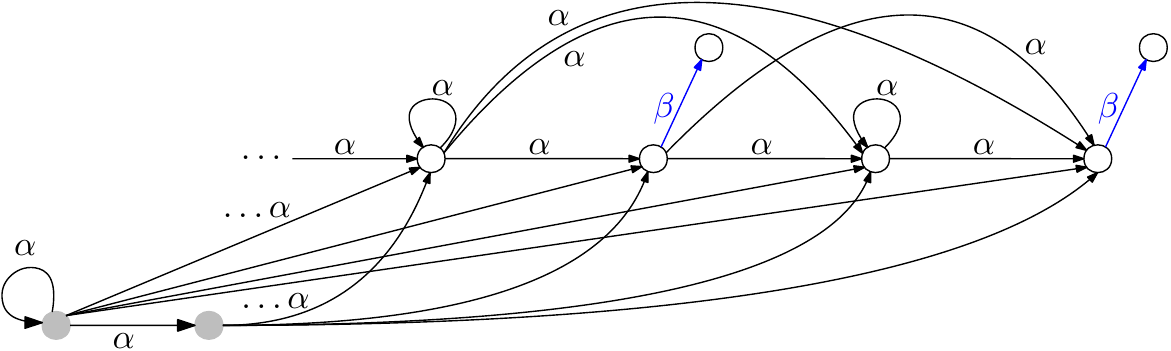}
		\caption{Kripke structure for Theorem \ref{ex:non_self_bisimilar_ks}. $L = \{\alpha, \beta, \alpha.\beta\}$; all black arrows have label $\alpha$; $\alpha.\beta$ transitions are not shown but span each concatenated $\alpha$ transition and $\beta$ transition.}
		\label{fig:non_self_bisimilar_ks}
	\end{figure}
	The proof of the claim in Example \ref{ex:non_self_bisimilar_ks} is 
	nontrivial, and as far as we know, there are no results even suggesting 
	that such Kripke structures exist. Importantly, Example 
	\ref{ex:non_self_bisimilar_ks} implies a similar example for GSTs because 
	it contains a Kripke structure that also satisfies the schemata in Theorem 
	\ref{thm:smaller_vhhm_classes_gsts}. The following example makes this 
	explicit.
	\begin{example}
	\label{ex:non_self_bisimilar_gst}
		Recall the definition of $G_A$ from Example 
		\ref{ex:modally_equiv_gsts_not_strongly_bisimilar}, and consider the 
		GST $G_X$ for $X = \{1/2 + 1/(n+2) : n \in \mathbb{N} \} \subset 
		(0,1)$. Then $G_X$ doesn't belong to any VHHM class of GSTs because its 
		surrogate Kripke structure is bisimilar to the Kripke structure in 
		Example \ref{ex:non_self_bisimilar_ks} (when it is suitably 
		relabeled).
	\end{example}


\section{Conclusions and Future Work} 
\label{sec:conclusions_and_future_work}
	In this paper we have proposed a generalization of Hennessy-Milner logic 
	that is suitable for GSTs, and we have used this logic to exhibit some 
	results regarding Hennessy-Milner classes with respect to weak 
	bisimulation. Nevertheless, there is a great deal of work still to be done. 
	One key avenue of future research lies in deciding whether the 
	characterization in Theorem \ref{thm:smaller_vhhm_classes_gsts} really 
	describes all maximal VHHM classes of GSTs (given a sufficiently large set 
	of GSTs to begin with). Another important avenue of future work is to 
	investigate what implications these VHHM classes have for common hybrid 
	system models, such as the behavioral modeling framework of 
	\cite{willems_behavioral_2007}.
\bibliographystyle{eptcs} %
\bibliography{mybib}

\begin{thebibliography}{10}
\providecommand{\bibitemdeclare}[2]{}
\providecommand{\surnamestart}{}
\providecommand{\surnameend}{}
\providecommand{\urlprefix}{Available at }
\providecommand{\url}[1]{\texttt{#1}}
\providecommand{\href}[2]{\texttt{#2}}
\providecommand{\urlalt}[2]{\href{#1}{#2}}
\providecommand{\doi}[1]{doi:\urlalt{http://dx.doi.org/#1}{#1}}
\providecommand{\bibinfo}[2]{#2}

\bibitemdeclare{incollection}{budde_theory_2014}
\bibitem{budde_theory_2014}
\bibinfo{author}{Carlos~E. \surnamestart Budde\surnameend},
  \bibinfo{author}{Pedro~R. \surnamestart D'Argenio\surnameend},
  \bibinfo{author}{Pedro~S{\'a}nchez \surnamestart Terraf\surnameend} \&
  \bibinfo{author}{Nicol{\'a}s \surnamestart Wolovick\surnameend}
  (\bibinfo{year}{2014}): \emph{\bibinfo{title}{A Theory for the Semantics of
  Stochastic and Non-deterministic Continuous Systems}}.
\newblock In: {\sl \bibinfo{booktitle}{Stochastic Model Checking. Rigorous
  Dependability Analysis Using Model Checking Techniques for Stochastic
  Systems}}, \bibinfo{series}{Lecture Notes in Computer Science},
  \bibinfo{publisher}{Springer, Berlin, Heidelberg}, pp.
  \bibinfo{pages}{67--86}, \doi{10.1007/978-3-662-45489-3\_3}.

\bibitemdeclare{inproceedings}{cleaveland1990automatically}
\bibitem{cleaveland1990automatically}
\bibinfo{author}{Rance \surnamestart Cleaveland\surnameend}
  (\bibinfo{year}{1990}): \emph{\bibinfo{title}{On automatically explaining
  bisimulation inequivalence}}.
\newblock In: {\sl \bibinfo{booktitle}{International Conference on Computer
  Aided Verification}}, \bibinfo{organization}{Springer}, pp.
  \bibinfo{pages}{364--372}, \doi{10.1007/BFb0023750}.

\bibitemdeclare{article}{dargenio_bisimulations_2012}
\bibitem{dargenio_bisimulations_2012}
\bibinfo{author}{Pedro~R. \surnamestart D'Argenio\surnameend} \&
  \bibinfo{author}{Pedro~S{\'a}nchez \surnamestart Terraf\surnameend}
  (\bibinfo{year}{2012}): \emph{\bibinfo{title}{Bisimulations for
  non-deterministic labelled Markov processes}}.
\newblock {\sl \bibinfo{journal}{Mathematical Structures in Computer Science}}
  \bibinfo{volume}{22}(\bibinfo{number}{1}), pp. \bibinfo{pages}{43--68},
  \doi{10.1017/S0960129511000454}.

\bibitemdeclare{incollection}{davoren_simulations_2007}
\bibitem{davoren_simulations_2007}
\bibinfo{author}{J~M \surnamestart Davoren\surnameend} \&
  \bibinfo{author}{Paulo \surnamestart Tabuada\surnameend}
  (\bibinfo{year}{2007}): \emph{\bibinfo{title}{On {Simulations} and
  {Bisimulations} of {General} {Flow} {Systems}}}.
\newblock In: {\sl \bibinfo{booktitle}{Hybrid {Systems}: {Computation} and
  {Control}}}, \bibinfo{publisher}{Springer Berlin Heidelberg},
  \bibinfo{address}{Berlin, Heidelberg}, pp. \bibinfo{pages}{145--158},
  \doi{10.1007/978-3-540-71493-4\_14}.

\bibitemdeclare{article}{desharnais_bisimulation_2002}
\bibitem{desharnais_bisimulation_2002}
\bibinfo{author}{Jos{\'e}e \surnamestart Desharnais\surnameend},
  \bibinfo{author}{Abbas \surnamestart Edalat\surnameend} \&
  \bibinfo{author}{Prakash \surnamestart Panangaden\surnameend}
  (\bibinfo{year}{2002}): \emph{\bibinfo{title}{Bisimulation for {Labelled}
  {Markov} {Processes}}}.
\newblock {\sl \bibinfo{journal}{Information and Computation}}
  \bibinfo{volume}{179}(\bibinfo{number}{2}), pp. \bibinfo{pages}{163--193},
  \doi{10.1006/inco.2001.2962}.

\bibitemdeclare{article}{doberkat_stochastic_2005}
\bibitem{doberkat_stochastic_2005}
\bibinfo{author}{E.~\surnamestart Doberkat\surnameend} (\bibinfo{year}{2005}):
  \emph{\bibinfo{title}{Stochastic {Relations}: {Congruences}, {Bisimulations}
  and the {Hennessy}--{Milner} {Theorem}}}.
\newblock {\sl \bibinfo{journal}{SIAM Journal on Computing}}
  \bibinfo{volume}{35}(\bibinfo{number}{3}), pp. \bibinfo{pages}{590--626},
  \doi{10.1137/S009753970444346X}.

\bibitemdeclare{article}{doberkat_hennessymilner_2007}
\bibitem{doberkat_hennessymilner_2007}
\bibinfo{author}{Ernst-Erich \surnamestart Doberkat\surnameend}
  (\bibinfo{year}{2007}): \emph{\bibinfo{title}{The {Hennessy}--{Milner}
  equivalence for continuous time stochastic logic with mu-operator}}.
\newblock {\sl \bibinfo{journal}{Journal of Applied Logic}}
  \bibinfo{volume}{5}(\bibinfo{number}{3}), pp. \bibinfo{pages}{519--544},
  \doi{10.1016/j.jal.2006.05.001}.

\bibitemdeclare{article}{doberkat_stochastic_2017}
\bibitem{doberkat_stochastic_2017}
\bibinfo{author}{Ernst-Erich \surnamestart Doberkat\surnameend} \&
  \bibinfo{author}{Pedro \surnamestart S{\'a}nchez~Terraf\surnameend}
  (\bibinfo{year}{2017}): \emph{\bibinfo{title}{Stochastic non-determinism and
  effectivity functions}}.
\newblock {\sl \bibinfo{journal}{Journal of Logic and Computation}}
  \bibinfo{volume}{27}(\bibinfo{number}{1}), pp. \bibinfo{pages}{357--394},
  \doi{10.1093/logcom/exv049}.

\bibitemdeclare{incollection}{ferlez_generalized_2014}
\bibitem{ferlez_generalized_2014}
\bibinfo{author}{James \surnamestart Ferlez\surnameend}, \bibinfo{author}{Rance
  \surnamestart Cleaveland\surnameend} \& \bibinfo{author}{Steve \surnamestart
  Marcus\surnameend} (\bibinfo{year}{2014}): \emph{\bibinfo{title}{Generalized
  {Synchronization} {Trees}}}.
\newblock In: {\sl \bibinfo{booktitle}{Foundations of {Software} {Science} and
  {Computation} {Structures} ({FoSSaCS})}}, \bibinfo{publisher}{Springer Berlin
  Heidelberg}, \bibinfo{address}{Berlin, Heidelberg}, pp.
  \bibinfo{pages}{304--319}, \doi{10.1007/978-3-642-54830-7\_20}.

\bibitemdeclare{book}{goldblatt_logics_1992}
\bibitem{goldblatt_logics_1992}
\bibinfo{author}{Robert \surnamestart Goldblatt\surnameend}
  (\bibinfo{year}{1992}): \emph{\bibinfo{title}{Logics of {Time} and
  {Computation}}}, \bibinfo{edition}{{S}econd} edition.
\newblock \bibinfo{publisher}{CSLI Lecture Notes}, \bibinfo{address}{Stanford
  University, Center for the Study of Language and Information}.
\newblock
  \urlprefix\url{http://www.press.uchicago.edu/ucp/books/book/distributed/L/bo3615704.html}.

\bibitemdeclare{incollection}{goldblatt_saturation_1995}
\bibitem{goldblatt_saturation_1995}
\bibinfo{author}{Robert \surnamestart Goldblatt\surnameend}
  (\bibinfo{year}{1995}): \emph{\bibinfo{title}{Saturation and the
  {Hennessy}-{Milner} {Property}}}.
\newblock In \bibinfo{editor}{Alban \surnamestart Ponse\surnameend},
  \bibinfo{editor}{Maarten \surnamestart de~Rijke\surnameend} \&
  \bibinfo{editor}{Yde \surnamestart Venema\surnameend}, editors: {\sl
  \bibinfo{booktitle}{Modal {Logic} and {Process} {Algebra}: {A} {Bisimulation}
  {Perspective}}}, \bibinfo{series}{Center for the {Study} of {Language} and
  {Information} {Publication} {Lecture} {Notes}}, \bibinfo{publisher}{Cambridge
  University Press}, pp. \bibinfo{pages}{107--129}.

\bibitemdeclare{incollection}{goranko_model_2007}
\bibitem{goranko_model_2007}
\bibinfo{author}{Valentin \surnamestart Goranko\surnameend} \&
  \bibinfo{author}{Martin \surnamestart Otto\surnameend}
  (\bibinfo{year}{2007}): \emph{\bibinfo{title}{Model theory of modal logic}}.
\newblock In \bibinfo{editor}{Johan Van Benthem {and} Frank~Wolter
  \surnamestart Patrick~Blackburn\surnameend}, editor: {\sl
  \bibinfo{booktitle}{Studies in {Logic} and {Practical} {Reasoning}}}, {\sl
  \bibinfo{series}{Handbook of {Modal} {Logic}}}~\bibinfo{volume}{3},
  \bibinfo{publisher}{Elsevier}, pp. \bibinfo{pages}{249--329},
  \doi{10.1016/S1570-2464(07)80008-5}.

\bibitemdeclare{article}{hennessy_algebraic_1985}
\bibitem{hennessy_algebraic_1985}
\bibinfo{author}{Matthew \surnamestart Hennessy\surnameend} \&
  \bibinfo{author}{Robin \surnamestart Milner\surnameend}
  (\bibinfo{year}{1985}): \emph{\bibinfo{title}{Algebraic laws for
  nondeterminism and concurrency}}.
\newblock {\sl \bibinfo{journal}{Journal of the ACM}}, \doi{10.1145/2455.2460}.

\bibitemdeclare{incollection}{hollenberg_hennessy-milner_1995}
\bibitem{hollenberg_hennessy-milner_1995}
\bibinfo{author}{Marco \surnamestart Hollenberg\surnameend}
  (\bibinfo{year}{1995}): \emph{\bibinfo{title}{Hennessy-{Milner} {Classes} and
  {Process} {Algebra}}}.
\newblock In \bibinfo{editor}{Alban \surnamestart Ponse\surnameend},
  \bibinfo{editor}{Maarten \surnamestart de~Rijke\surnameend} \&
  \bibinfo{editor}{Yde \surnamestart Venema\surnameend}, editors: {\sl
  \bibinfo{booktitle}{Modal {Logic} and {Process} {Algebra}: {A} {Bisimulation}
  {Perspective}}}, \bibinfo{series}{Center for the {Study} of {Language} and
  {Information} {Publication} {Lecture} {Notes}}, \bibinfo{publisher}{Cambridge
  University Press}, pp. \bibinfo{pages}{187--216}.

\bibitemdeclare{book}{jech_set_2003}
\bibitem{jech_set_2003}
\bibinfo{author}{Thomas \surnamestart Jech\surnameend} (\bibinfo{year}{2003}):
  \emph{\bibinfo{title}{Set {Theory}}}, \bibinfo{edition}{{T}hird {M}illenium}
  edition.
\newblock \bibinfo{series}{Springer {Monographs} in {Mathematics}},
  \bibinfo{publisher}{Springer-Verlag Berlin Heidelberg},
  \doi{10.1007/3-540-44761-X}.

\bibitemdeclare{incollection}{milner_handbook_1990}
\bibitem{milner_handbook_1990}
\bibinfo{author}{Robin \surnamestart Milner\surnameend} (\bibinfo{year}{1990}):
  \emph{\bibinfo{title}{Operational and Algebraic Semantics of Concurrent
  Processes}}.
\newblock In: {\sl \bibinfo{booktitle}{Handbook of {Theoretical} {Computer}
  {Science} ({Vol}. {B})}}, \bibinfo{publisher}{MIT Press},
  \bibinfo{address}{Cambridge, MA, USA}, pp. \bibinfo{pages}{1201--1242},
  \doi{10.1016/B978-0-444-88074-1.50024-X}.

\bibitemdeclare{article}{terraf_bisimilarity_2015}
\bibitem{terraf_bisimilarity_2015}
\bibinfo{author}{Pedro~S{\'a}nchez \surnamestart Terraf\surnameend}
  (\bibinfo{year}{2015}): \emph{\bibinfo{title}{Bisimilarity is not Borel}}.
\newblock {\sl \bibinfo{journal}{Mathematical Structures in Computer Science}},
  pp. \bibinfo{pages}{1--20}, \doi{10.1017/S0960129515000535}.

\bibitemdeclare{article}{willems_behavioral_2007}
\bibitem{willems_behavioral_2007}
\bibinfo{author}{Jan \surnamestart Willems\surnameend} (\bibinfo{year}{2007}):
  \emph{\bibinfo{title}{The {Behavioral} {Approach} to {Open} and
  {Interconnected} {Systems}}}.
\newblock {\sl \bibinfo{journal}{IEEE Control Systems Magazine}}
  \bibinfo{volume}{27}(\bibinfo{number}{6}), pp. \bibinfo{pages}{46--99},
  \doi{10.1109/MCS.2007.906923}.

\end{thebibliography}

\end{document}